\documentclass[11pt]{article}
\usepackage[IL2]{fontenc}
\usepackage[letterpaper,margin=1.00in]{geometry}
\usepackage{amsmath, amssymb, amsthm, amsfonts}
\usepackage{bbm}
\usepackage{amscd}
\usepackage{mathrsfs}

\usepackage{comment} 
\usepackage{ifthen}
\usepackage{tikz}
\usetikzlibrary{positioning,decorations.pathreplacing}
\usepackage{ bbold }
\usepackage{xcolor}
\usepackage{appendix}
\usepackage{graphicx}
\usepackage{color}
\usepackage{algorithm}
\usepackage[noend]{algpseudocode}
\usepackage{epstopdf}
\usepackage{wrapfig}
\usepackage{paralist}
\usepackage{wasysym}
\usepackage[textsize=tiny]{todonotes}

\usepackage{listings} 

\usepackage{pdflscape} 

\usepackage{framed}
\usepackage[framemethod=tikz]{mdframed}
\usepackage[bottom]{footmisc}
\usepackage{enumitem}
\setitemize{noitemsep,topsep=3pt,parsep=3pt,partopsep=3pt}
\usepackage[font=small]{caption}
\usepackage{xspace}

\usepackage{thmtools} 
\usepackage{thm-restate} 

\usepackage{hyperref}
\hypersetup{
    unicode=false,          
    colorlinks=true,        
    linkcolor=red,          
    citecolor=darkgreen,        
    filecolor=magenta,      
    urlcolor=cyan           
}

\newtheorem{theorem}{Theorem}[section]

\newtheorem{meta-theorem}[theorem]{Meta-Theorem}
\newtheorem{claim}[theorem]{Claim}

\newtheorem{corollary}[theorem]{Corollary}

\newtheorem{definition}[theorem]{Definition}

\usepackage[capitalize, nameinlink,noabbrev]{cleveref}

\crefname{theorem}{Theorem}{Theorems}
\crefname{proposition}{Proposition}{Propositions}
\crefname{observation}{Observation}{Observations}
\crefname{lemma}{Lemma}{Lemmas}
\crefname{claim}{Claim}{Claims}
\crefname{problem}{Problem}{Problems}
\crefname{conjecture}{Conjecture}{Conjectures}
\crefname{question}{Question}{Questions}
\crefname{example}{Example}{Examples}
\crefname{fact}{Fact}{Facts}
\crefname{figure}{Figure}{Figures}
\crefname{table}{Table}{Tables}

\definecolor{darkgreen}{rgb}{0,0.5,0}

\algnewcommand\algorithmicswitch{\textbf{switch}}
\algnewcommand\algorithmiccase{\textbf{case}}

\algdef{SE}[SWITCH]{Switch}{EndSwitch}[1]{\algorithmicswitch\ #1\ \algorithmicdo}{\algorithmicend\ \algorithmicswitch}%
\algdef{SE}[CASE]{Case}{EndCase}[1]{\algorithmiccase\ #1}{\algorithmicend\ \algorithmiccase}%
\algtext*{EndSwitch}%
\algtext*{EndCase}%

\newcommand{\fC}{\mathcal{C}}

\usepackage{color, colortbl}
\definecolor{Gray}{gray}{0.9}

\usepackage{wasysym}

\newcommand{\congest}{$\mathsf{CONGEST}\,$}
\newcommand{\local}{$\mathsf{LOCAL}\,$}

\newcommand{\poly}{\operatorname{poly}}

\renewcommand{\phi}{\varphi}

\renewcommand{\paragraph}[1]{\vspace{0.15cm}\noindent {\bf #1}:}

\newcommand{\FullOrShort}{full}
\ifthenelse{\equal{\FullOrShort}{full}}{
  
  \newcommand{\fullOnly}[1]{#1}
  \newcommand{\shortOnly}[1]{}

  }{

    \newcommand{\fullOnly}[1]{}
    \newcommand{\IncludePictures}[1]{}
   
  }

\title{A Simple Deterministic Distributed \\
Low-Diameter Clustering}

\author{
Václav Rozhoň \thanks{\scriptsize{Supported by the European Research Council (ERC) under the European Unions Horizon 2020 research and innovation programme (grant agreement No.~853109).}} \\
\small ETH Zurich \\
\small rozhonv@inf.ethz.ch\\
\and \textcircled{r}\footnote{The author ordering was randomized using \url{https://www.aeaweb.org/journals/policies/random-author-order/generator}. 
 It is requested that citations of this work list the authors separated by \texttt{\textbackslash textcircled\{r\}} instead of commas} 
 \and
Bernhard Haeupler \thanks{\scriptsize{Supported in part by NSF grants CCF-1814603, CCF-1910588, NSF CAREER award CCF-1750808, a Sloan Research Fellowship, funding from the European Research Council (ERC) under the European Union's Horizon 2020 research and innovation program (ERC grant agreement 949272), and the Swiss National Foundation (project grant 200021-184735).}}\\
\small Carnegie Mellon University \& ETH Zurich\\
\small bernhard.haeupler@inf.ethz.ch\\
\and \textcircled{r} \and
Christoph Grunau \footnotemark[1]\\
\small ETH Zurich \\
\small cgrunau@inf.ethz.ch\\
}

\begin{document}

\maketitle

\begin{abstract}
We give a simple, local process for nodes in an undirected graph to form non-adjacent clusters that (1) have at most a polylogarithmic diameter and (2) contain at least half of all vertices.

\smallskip
Efficient\footnote{With efficient we mean polylogarithmic rounds in \congest~\cite{peleg00}, i.e., the standard model for distributed message-passing algorithms.} deterministic distributed clustering algorithms for computing strong-diameter network decompositions and other key tools follow immediately.
Overall, our process is a direct and drastically simplified way for computing these fundamental objects.
\end{abstract}

\section{Introduction}

This paper focuses on distributed graph algorithms, particularly on the fundamental problem of deterministic and local ways to compute network decompositions and \emph{low-diameter clusterings}, which cluster at least half of the nodes in a given graph into non-adjacent clusters with small diameter. In particular, the paper describes a drastically simplified efficient deterministic distributed construction for computing such a low-diameter clustering with polylogarithmic diameter in polylogarithmic rounds of the distributed \congest model. 


\medskip
Starting with the seminal work of Luby \cite{luby86_lubys_alg} from the 1980's, fast and simple $O(\log n)$-round \emph{randomized} distributed algorithms are known for many fundamental symmetry breaking problems like maximal independent set (MIS) or $\Delta+1$ vertex coloring. 
For a long time, this was in stark contrast with the state-of-the-art deterministic algorithms. 
For multiple decades, it was a major open problem in the area of distributed graph algorithms to get deterministic algorithms with round complexity $\poly \log (n)$ for such problems, e.g., MIS or $\Delta+1$ vertex coloring. A recent breakthrough of Rozhoň and Ghaffari  \cite{rozhon_ghaffari2019decomposition} managed to resolve this open problem. 

In their work, the authors presented the first polylogarithmic-round deterministic algorithm for network decompositions using a (weak-diameter version of) low-diameter clusterings. Network decomposition is the object we get by repeatedly finding a low diameter clustering and removing all the nodes in the clustering, until no node remains. See \cref{subsec:background} for the formal definitions. It was long known that low-diameter clusterings is the up-to-then-missing fundamental tool required for a large class of \local deterministic distributed algorithms. The clustering construction of \cite{rozhon_ghaffari2019decomposition} directly implied, among others, first efficient distributed algorithms for MIS (together with the work of \cite{censor2017deterministic_mis_congest}) and $\Delta+1$ vertex coloring (together with the work of \cite{bamberger2020efficient}) in the standard bandwidth-limited \congest model of distributed computing.

The main difference in the natural low-diameter clustering problem defined above and the weaker version solved in \cite{rozhon_ghaffari2019decomposition} is that clusters are not necessarily connected or induce a low low-diameter subgraph on their own but instead have low \emph{weak-diameter}. A cluster has weak-diameter at most $D$ if any two nodes in the cluster are connected by a path of length at most $D$ \emph{in the original graph $G$ instead of within the cluster itself}. Hence, a cluster may even be disconnected. While the weak-diameter guarantee is enough for derandomizing local computations without bandwidth limitations, including MIS and $\Delta+1$-coloring, the original -- strong-diameter -- clustering stated above is clearly the natural and right object to ask for: It is strictly stronger, easier to define, easier to use in applications, and requires less and simpler objects and notation. Indeed, in distributed models with bandwidth limitations, such as the standard \congest model in which message sizes are restricted, it is not sufficient that clusters have small weak-diameter but one also needs to guarantee that there exist so-called low-depth Steiner trees connecting the nodes of each cluster. The collection of these Steiner-trees must furthermore satisfy additional low-congestion guarantees, i.e., each edge or each node in the graph is not used by too many trees (as a Steiner node). Algorithms must also be able to compute the Steiner forest of a weak-diameter clustering efficiently. Lastly, there are several applications, e.g., low-stretch spanning trees, where strong-diameter clusterings are strictly required and the weak-diameter guarantee does not suffice \cite{elkin_haeupler_rozhon_gruanu2022Clusterings_LSST}. 
This motivated the later works of \cite{chang_ghaffari2021strong_diameter,elkin_haeupler_rozhon_gruanu2022Clusterings_LSST} to give low-diameter clustering algorithms with strong-diameter guarantees, typically first building a weak-diameter clustering and then using this weak-diameter clustering either for communication or using it as a starting point for building a strong-diameter clustering out of it recursively. This multi-step process still requires to define and maintain Steiner forests for weak-diameter clusterings during intermediate steps. 

In this work, we show that there is a much simpler and direct way to get strong-diameter guarantees by designing a natural clustering process that combines key ideas from \cite{rozhon_ghaffari2019decomposition} and \cite{elkin_haeupler_rozhon_gruanu2022Clusterings_LSST}.

\subsection{Preliminaries: Distributed \congest Model and Low-Diameter Clusterings}
\label{subsec:background}

We will now briefly introduce the standard model for distributed message-passing algorithms -- the \congest model of distributed computing~\cite{peleg00} and also give the definitions of clustering that we use (see \cite{elkin_haeupler_rozhon_gruanu2022Clusterings_LSST} for more discussion). 

\paragraph{\congest}
Throughout the paper, we work with the \congest model, which is the standard distributed message-passing model for graph algorithms~\cite{peleg00}. 
The network is abstracted as an $n$-node undirected graph $G=(V, E)$ where each node $v\in V$ corresponds to one processor in the network. Communications take place in synchronous rounds. Per round, each node sends one $O(\log n)$-bit message to each of its neighbors in $G$. 
We also consider the relaxed variant of the model where we allow unbounded message sizes, called \local. At the end of the round, each node performs some computations on the data it holds, before we proceed to the next communication round. 

We capture any graph problem in this model as follows: Initially, the network topology is not known to the nodes of the graph, except that each node $v\in V$ knows its own unique $O(\log n)$-bit identifier. It also knows a suitably tight (polynomial) upper bound on the number $n$ of nodes in the network. At the end of the computation, each node should know its own part of the output, e.g., in the graph coloring problem, each node should know its own color. 

Whenever we say that there is ``an efficient distributed algorithm'', we mean that there is a \congest algorithm for the problem with round complexity $\poly(\log n)$. 

\paragraph{Low Diameter Clustering}
The main object of interest that we want to construct is a so-called \emph{low diameter clustering}, which we formally define after introducing a bit of notation. Throughout the whole paper we work with undirected unweighted graphs and write $G[U]$ for the subgraph of $G$ induced by $U \subseteq V(G)$. We use $d_G(u,v)$ to denote the distance of two nodes $u,v \in V(G)$ in $G$. We also simplify the notation to $d(u,v)$ when $G$ is clear from context and generalize it to sets by defining $d_G(U,W) = \min_{u \in U, w \in W} d_G(u,w)$ for $U,W \subseteq V(G)$. The diameter of $G$ is defined as $\max_{u,v \in V(G)} d_G(u,v)$. 

We use the term \emph{clustering of $G$} to denote any set of disjoint vertex subsets of $G$. A low diameter clustering is a clustering with additional properties: 

\begin{definition}[Low Diameter Clustering]
\label{def:low_diam_clustering}
A low diameter clustering $\fC$ with diameter $D$ of a graph $G$ is a clustering of $G$ such that:
\begin{enumerate}
    \item No two clusters $C_1 \not= C_2 \in \fC$ are adjacent in $G$, i.e., $d(C_1, C_2) \ge 2$. 
    \item For every cluster $C \in \fC$, the diameter of $G[C]$ is at most $D$.
\end{enumerate}
\end{definition}

Similarly, we define a low diameter clustering with weak-diameter at most $D$ by replacing the condition (2) with he requirement that for each cluster $C \in \fC$ and any two nodes $u,v \in C$ we have $d_G(u,v) \le D$. 

Whenever we construct a low diameter clustering, we additionally want it to cover as many nodes as possible. Usually, we want to cover at least half of the nodes of $G$, or formally, we require that $\left| \bigcup_{C \in \fC} C \right| \ge n/2$. 
Sometimes, it is also necessary to generalize (1) and require a larger separation of the clusters, but this is not considered in this paper. 

Let us now give a formal definition of network decomposition.

\begin{definition}[Network Decomposition]
\label{def:network_decomposition}
A network decomposition with $C$ colors and diameter $D$ is a coloring of nodes with colors $1, 2, \dots, C$ such that each color induces a low-diameter clustering of diameter $D$. 
\end{definition}

Notice that whenever we can construct a low-diameter clustering with diameter $D$ that covers at least $n/2$ nodes, we get a network decomposition by repeatedly constructing a low diameter clustering and removing it from the graph. This way, we achieve a network decomposition with $C = O(\log n)$ and diameter $D$. 
Since virtually all deterministic constructions of network decomposition work this way, we focus on constructing low-diameter clusterings from now on. 

The reason why network decomposition is a useful object is that it corresponds to the canonical way of using clusterings in distributed computing.
To give an example, we show how to use it to solve the maximal independent set problem in the less restrictive \local model. 

Given access to a network decomposition, we iterate over the $C$ color classes and gradually build independent sets $I_1 \subseteq I_2 \subseteq \dots \subseteq I_C$ where $I_C$ is maximal. 
In the $i$-th step, each cluster $K$ of the low-diameter clustering induced by the $i$-th color computes a maximal independent set in the graph induced by all the nodes in $K$ that are not neighboring a node in $I_{i-1}$ and we define $I_i$ by adding these independent sets to $I_{i-1}$. The set $I_C$ is clearly maximal. 
Computing the maximal independent set inside one cluster $K$ can be done in $O(D)$ rounds of the \local model as follows: One node of the cluster collects all the information about $G[K]$ and its neighborhood in $G$, then locally computes a maximal independent set, and afterwards broadcasts the solution to the nodes in the cluster.  Hence, the overall algorithm has round complexity $O(CD)$. Hence, given a network decomposition with $C,D = \poly(\log n)$, one can compute a maximal independent set in $\poly(\log n)$ rounds. 
Note that this brute-force approach for computing a maximal independent set critically relies on the fact that the \local model does not restrict the size of messages.

In the more restrictive \congest model, computing a maximal independent set inside a low diameter cluster becomes nontrivial, but one can use the deterministic MIS algorithm of \cite{censor2017deterministic_mis_congest} with round complexity $O(D \cdot \poly(\log n))$ where $D$ is the diameter of the input graph.

\subsection{Comparison with Previous Work}

We summarize the work on deterministic distributed low-diameter clusterings in the \congest model in \cref{table:papers}. 

\begin{table}[ht]
\centering
\begin{tabular}{||c | p{2.7cm} | c | p{2cm} |  c||} 

 \hline
 Paper & Fraction of clustered nodes & Diameter of clusters  & Strong diameter? &  round complexity \\ [0.5ex] 
 \hline\hline
 \cite{awerbuch89} & $2^{- \Omega(\sqrt{\log n \log \log n})}$& $2^{O(\sqrt{\log n \log \log n})}$ & \checkmark & $2^{O(\sqrt{\log n \log \log n})}$ \\ 
  \hline
 \cite{ghaffari2019distributed_MIS_congest} & $2^{-\Omega(\sqrt{\log n})}$& $2^{O(\sqrt{\log n})}$ & \checkmark & $2^{O(\sqrt{\log n })}$ \\ 
  
  \rowcolor{Gray}
  \hline
\cite{rozhon_ghaffari2019decomposition} & $1/2$ & $O(\log^3 n)$ & $\times$ & $O(\log^7 n)$ \\ 
  \hline
  \cite{ghaffari_grunau_rozhon2020improved_network_decomposition} & $1/2$ & $O(\log^2 n)$ & $\times$ &  $O(\log^4 n)$ \\ 
  \hline
  \cite{chang_ghaffari2021strong_diameter} & $1/2$ & $O(\log^2 n)$ & \checkmark &  $O(\log^{10} n)$ \\ 
  \hline
    \cite{chang_ghaffari2021strong_diameter} & $1/2$ & $O(\log^3 n)$ & \checkmark &  $O(\log^7 n)$ \\ 
\rowcolor{Gray}  \hline
\cite{elkin_haeupler_rozhon_gruanu2022Clusterings_LSST} & $1/2$ & $O(\log^2 n)$ & \checkmark & $O(\log^4 n)$ \\ 
  \hline
\cite{ghaffari2022improved} & $\Omega(1 / \log\log n)$ & $O(\log n ) $ & \checkmark &  $\log^2(n) \cdot \poly(\log\log n)$ \\ 
  \hline
\cite{ghaffari2022improved} & $1/2$ & $O(\log n \, \cdot \, \log\log\log n) $ & \checkmark &  $\log^2(n) \cdot \poly(\log\log n)$ \\ 
\rowcolor{Gray}
  \hline
 this paper & $1/2$ & $O(\log^3 n) $ & \checkmark &  $\log^6(n)$ \\ 
  \hline
 
\end{tabular}

\caption{This table shows the previous work on distributed deterministic algorithms for low-diameter clusterings. We highlighted the three results relevant for this paper. }
 \label{table:papers}

\end{table}

There are three highlighted rows in the table, besides our result we highlight the work of \cite{rozhon_ghaffari2019decomposition} and \cite{elkin_haeupler_rozhon_gruanu2022Clusterings_LSST}; The algorithm of this paper combines ideas from both of these papers. 

Let us now go through the rows of the table. The first two rows, together with the related results of \cite{panconesi-srinivasan,ghaffari_portmann2019improved,ghaffari_kuhn2018derandomizing_spanners_dominatingsets} represent the results before the work of \cite{rozhon_ghaffari2019decomposition} and are not relevant to our paper. 

Next, there is the work of \cite{rozhon_ghaffari2019decomposition} and an improved variant of it by \cite{ghaffari_grunau_rozhon2020improved_network_decomposition}. These were the first deterministic efficient constructions of low diameter clusterings, however, they suffer from only providing a weak-diameter guarantee. 

Next, the work of \cite{chang_ghaffari2021strong_diameter} and \cite{elkin_haeupler_rozhon_gruanu2022Clusterings_LSST} use the algorithm of \cite{ghaffari_grunau_rozhon2020improved_network_decomposition} as a black blox and use additional ideas on top of the weak-diameter algorithm to create strong-diameter clusterings. The row with \cite{elkin_haeupler_rozhon_gruanu2022Clusterings_LSST} is highlighted because our algorithm uses an idea similar to theirs. 

Finally, a very recent algorithm of \cite{ghaffari2022improved} manages to bring down the diameter of the clusters as well as the round complexity, with a very different technique than \cite{rozhon_ghaffari2019decomposition}. 
However, their algorithm is very complicated. 

By far the simplest efficient algorithm from those in the table is the one from \cite{rozhon_ghaffari2019decomposition}. We show that with a small modification to their algorithm in the spirit of the algorithm of  \cite{elkin_haeupler_rozhon_gruanu2022Clusterings_LSST}, we can get a very simple algorithm computing strong-diameter clusters. 
Formally, we show the following result. 

\begin{theorem}
\label{thm:main}
There is a deterministic distributed algorithm that outputs a clustering $\fC$ of the input graph $G$ consisting of separated clusters of diameter $O(\log^3 n)$ such that at least $n/2$ nodes are clustered. 
The algorithm runs in $O(\log^6 n)$ \congest rounds. 
\end{theorem}

Recall that by repeatedly applying above result we get the following corollary. 

\begin{corollary}
\label{cor:main}
There is a deterministic distributed algorithm that outputs a network decomposition with $C = O(\log n)$ colors and diameter $D = O(\log^3 n)$. The algorithm runs in $O(\log^7 n)$ \congest rounds. 
\end{corollary}

\paragraph{Comparison of our algorithm with \cite{rozhon_ghaffari2019decomposition}}

We now give a high-level explanation of the algorithm of \cite{rozhon_ghaffari2019decomposition} and afterwards compare it to our algorithm.

In the algorithm of \cite{rozhon_ghaffari2019decomposition}, we start with a trivial clustering where every node is a cluster. Every cluster inherits the unique identifier from the starting node. During the algorithm, a cluster can grow, shrink and some vertices are deleted from the graph and will not be part of the final output clustering. In the end, the nonempty clusters cluster at least $n/2$ nodes and their weak-diameter is $O(\log^3 n)$. 

More concretely, the algorithm consists of $b = O(\log n)$ phases where $b$ is the number of bits in the node identifiers. 
In phase $i$, we split clusters into red and blue clusters based on the $i$-th bit in their identifier; the goal of the phase is to disconnect the red from the blue clusters by deleting at most $n/(2b)$ nodes in the graph. 

Here is how this is done. The $i$-th phase consists of $O(b \log n)$ steps. In general red clusters can only grow and blue clusters can only shrink.
More concretely, in each step every node in a blue cluster neighboring with a red cluster proposes to join an arbitrary neighboring red cluster. 
Now, for a given red cluster $C$, if the total number of proposing blue nodes is at least $|C|/(2b)$, then $C$ decides to grow by adding all the proposing blue nodes to the cluster. Otherwise, the proposing nodes are deleted which results in $C$ not being adjacent to any other blue nodes until the end of the phase. 

One can see that the number of deleted nodes per phase is only $n/(2b)$ in total, as needed. On the other hand, each cluster can grow only $O(b \log n)$ times until it has more than $n$ nodes, which implies that the weak-diameter of each cluster grows only by $O(b \log n) = O(\log^2 n)$ per phase. 

This concludes the description of the algorithm of \cite{rozhon_ghaffari2019decomposition}. Note that the clusters from their algorithm only have small weak-diameter since the nodes in a cluster can leave it in the future and the cluster may then even disconnect. 

\textbf{Our strong-diameter algorithm: }
To remedy the problem with the weak-diameter guarantee, we change the algorithm of \cite{rozhon_ghaffari2019decomposition} as follows: 
Instead of clusters, we will think in terms of their centers that we call \emph{terminals}. Given a set of terminals $Q$ such that $Q$ is $R$-ruling, i.e., for every $u \in V(G)$ we have $d_G(Q, u) \le R$, we can always construct a clustering with strong-diameter $R$ by running a breadth first search from $Q$. Hence, keeping a set of terminals is equivalent to keeping a set of strong-diameter clusters. 

Our algorithm starts with the trivial clustering where $Q = V(G)$. During the algorithm, we keep a set of terminals $Q$ and in each of the $b$ phases we delete at most $n/(2b)$ nodes and make some nodes of $Q$ nonterminals such that those remaining terminals with their $i$-th bit equal to $0$ are in a different component than those that have their $i$-th bit equal to $1$ (see \cref{fig:outer}). Moreover, we want that if at the beginning of the phase the set $Q$ is $R$-ruling, then it is $R + O(b \log n)$-ruling at the end of the phase (cf. the $O(b \log n)$ increase in weak-diameter in the algorithm of \cite{rozhon_ghaffari2019decomposition}). 

At the beginning of each phase, we run a breadth first search from the set $Q$, which gives us a clustering with strong diameter $R$ (see the left picture in \cref{fig:alg}). We in fact think of each cluster as a rooted tree of radius $R$. 

We then implement the same growing process as  \cite{rozhon_ghaffari2019decomposition}, but with a twist: whenever a blue node $v$ proposes to join a red cluster, the whole subtree rooted at $v$ proposes instead of just $v$ (see the middle picture in \cref{fig:alg}). 
This is because rehanging/deleting the whole subtree does not break the strong-diameter guarantee of blue clusters. 
If a blue node joins a red cluster, it stops being a terminal. 

The only new argument that needs to be done is that the diameter of red clusters does not grow a lot, which is trivial in the algorithm of \cite{rozhon_ghaffari2019decomposition} and follows by a simple argument in our algorithm. 

We note that the algorithm of \cite{elkin_haeupler_rozhon_gruanu2022Clusterings_LSST} also keeps track of terminals. However, to separate the red and blue terminals in one phase their algorithm relies on computing global aggregates, which can only be done efficiently on a low-diameter input graph.



\section{Clustering Algorithm}

In this section we prove \cref{thm:clustering_theorem} given below, which is a more precise version of \cref{thm:main}. 

\begin{theorem}[Clustering Theorem]
\label{thm:clustering_theorem}
Consider an arbitrary $n$-node network graph $G = (V,E)$ where each node has a unique $b = O(\log n)$-bit identifier. There is a deterministic distributed algorithm that, in $O(\log^6 n)$ rounds in the \congest model, finds a subset $V' \subseteq V$ of nodes, where $|V'| \geq |V|/2$, such that the subgraph $G[V']$ induced by the set $V'$ is partitioned into non-adjacent disjoint clusters of diameter $O(\log^3 n)$.
\end{theorem}

\begin{figure}[ht]
    \centering
    \includegraphics[width = .9\textwidth]{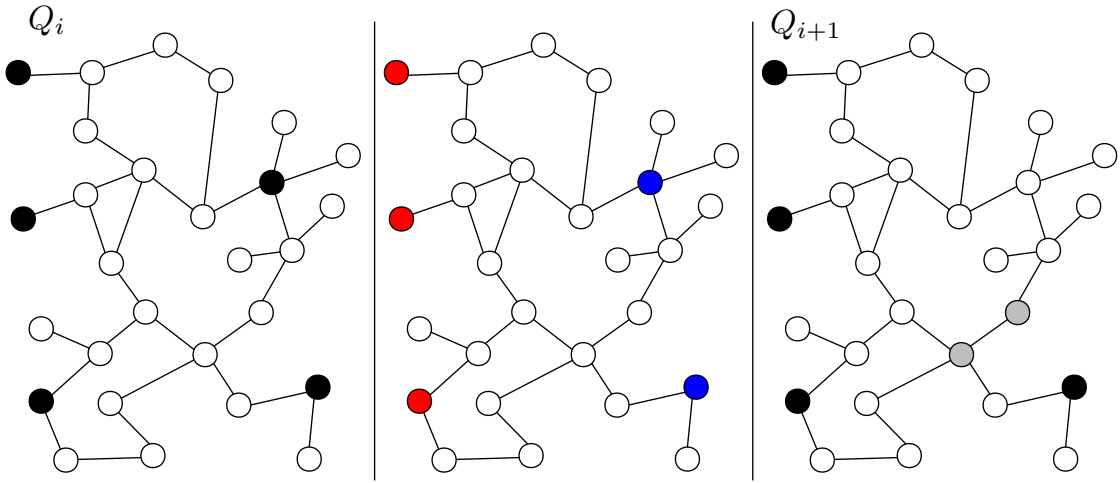}
    \caption{The figure shows one phase of the algorithm from \cref{thm:clustering_theorem}. The left figure contains a $3$-ruling set of terminal nodes $Q_i$ that we start with at the beginning of phase $i$. We split $Q_i$ into red and blue terminals according to the $(i+1)$-th bit of their identifiers. Then, we implement one phase of the algorithm. As a result, some of the nodes are deleted (grey) and some blue terminals stop being terminals. The set of remaining terminals $Q_{i+1}$ is on one hand $6$-ruling, on the other hand the blue terminals in $Q_{i+1}$ are separated from the red terminals. }
    \label{fig:outer}
\end{figure}

\begin{figure}[ht]
    \centering
    \includegraphics[width = .9\textwidth]{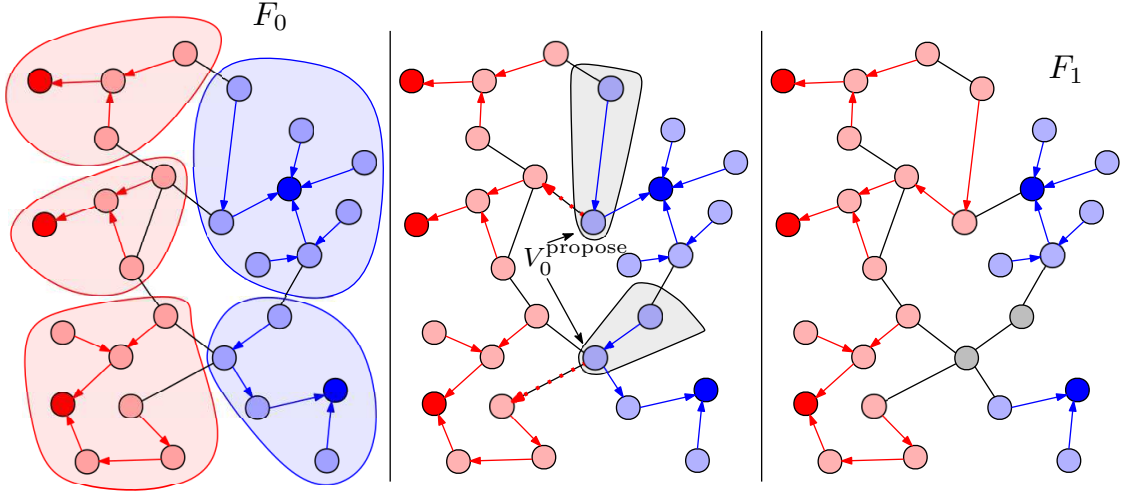}
    \caption{This figure explains one step of the algorithm of \cref{thm:clustering_theorem}, namely it shows what happens between the middle and the right picture of \cref{fig:outer}. 
    The left picture illustrates the beginning of the phase where we compute a BFS forest $F_0$ from the set $Q$ of terminals. In the first (and any other) step (the middle picture) we construct a set $V_0^{propose}$. Some proposals are accepted and the respective blue nodes join red clusters, while some proposals are rejected and respective blue nodes are deleted (the right picture).    }
    \label{fig:alg}
\end{figure}

We start by describing the algorithm outline of \cref{thm:clustering_theorem}. The construction has $b = O(\log n)$ phases, corresponding to the number of bits in the identifiers.
For $i \in [0,b-1]$, we denote by $V_i$ the set of living vertices at the beginning of phase $i$. 
Initially, all nodes are living and therefore $V_0 = V$.
In each phase, at most $|V|/(2b)$ nodes die. Dead nodes remain dead and will not be contained in $V'$.
Some of the alive nodes are terminals. We denote the set of terminals at the beginning of phase $i$ by $Q_i$.
Initially, all living nodes are terminals and therefore $Q_0 = V$.

Slightly abusing the notation, we let $V_b$ and $Q_b$ denote the set of living vertices and terminals at the end of phase $b-1$, respectively. We define $V'$ to be the final set of living nodes, i.e., $V' = V_b$, and each connected component of $G[V']$ will contain exactly one terminal in $Q_b$.

For stating the key invariants the algorithm satisfies, we need the following standard definition of a ruling set:

\begin{definition}[Ruling set]
We say that a subset $Q \subseteq V(G)$ is $R$-ruling in $G$ if every node $v \in V(G)$ satisfies $d_G(Q, v) \le R$. 
\end{definition}

\paragraph{Construction invariants}
The construction is such that, for each $i \in [0,b]$, the following three invariants are satisfied:

\begin{enumerate}[label = \Roman*.]
    \item Ruling Invariant: $Q_i$ is $R_i$-ruling in $G[V_i]$ for $R_i = i \cdot O(\log^2 n)$.
    \item Separation Invariant: Let $q_1,q_2 \in Q_i$ be two nodes in the same connected component of $G[V_i]$. Then, the identifiers of $q_1$ and $q_2$ coincide in the first $i$ bits.
    \item Deletion Invariant: $|V_i| \geq \left(1 - \frac{i}{2b}\right)|V|$.
\end{enumerate}

Note that setting $V_0 = Q_0 = V$ indeed results in the invariant being satisfied for $i = 0$.
In the end, we set $V' = V_b$. The deletion invariant for $i = b$ states that $|V'| \geq |V|/2$. The separation invariant implies that each connected component of $G[V']$ contains at most one node of $Q_b$. Together with the ruling invariant, which states that $Q_b$ is $R_b$-ruling in $G[V']$ for $R_b = O(\log^3 n)$, this implies that each connected component of $G[V']$ has diameter $O(\log^3 n)$. 
Next, in \cref{sec:outline_one_phase} we present the outline of one phase.
Afterwards, in \cref{sec:analysis_phase} we prove the correctness of the algorithm and analyse the \congest complexity.

\subsection{Outline of One Phase}
\label{sec:outline_one_phase}
In phase $i$, we compute a sequence of rooted forests $F_0, F_1, \ldots, F_t$ in $t = 2b^2 = O(\log^2 n)$ steps. 
At the beginning, $F_0$ is simply a BFS forest in $G[V_i]$ from the set $Q_i$. 
At the end, we set $V_{i+1} = V(F_t)$ and $Q_{i+1}$ is the set of roots of the forest $F_t$.

Let $j \in \{0,1,\ldots,t-1\}$ be arbitrary. We now explain how $F_{j+1}$ is computed given $F_j$. In general, each node contained in $F_{j+1}$ is also contained in $F_j$, i.e., $V(F_{j+1}) \subseteq V(F_j)$, and each root of $F_{j+1}$ is also a root in $F_j$. We say that a tree in $F_j$ is a red tree if the $(i+1)$-th bit of the identifier of its root is $0$  and otherwise we refer to the tree as a blue tree.
Also, we refer to a node in a red tree as a red node and a node in a blue tree as a blue node.
Each red node in $F_j$ will also be a red node in $F_{j+1}$. Moreover, the path to its root is the same in both $F_j$ and $F_{j+1}$.
Each blue node in $F_j$ can (1) either be a blue node in $F_{j+1}$, in which case the path to its root is the same in both $F_j$ and $F_{j+1}$, (2) be deleted and therefore not be part of any tree in $F_{j+1}$, (3) become a red node in $F_{j+1}$.

Let $V_j^{propose}$ be the set which contains each node $v$ which (1) is a blue node in $F_j$, and (2) $v$ is the only node neighboring a red node (in the graph $G$) in the path from $v$ to its root in $F_j$.
For a node $v \in V_j^{propose}$, let $T_v$ be the subtree rooted at $v$ with respect to $F_j$. Note that it directly follows from the way we defined $V_j^{propose}$ that $v$ is the only node in $T_v$ which is contained in $V_j^{propose}$.

Each node in $V_j^{propose}$ proposes to an arbitrary neighboring red tree in $F_j$. Now, a given red tree $T$ in $F_j$ decides to grow if 

\[\sum_{\substack{v \in V^{propose}_j \colon \\ \text{$v$ proposes to $T$}}} |V(T_v)| \geq \frac{|V(T)|}{2b}.\]

If $T$ decides to grow, then it accepts all the proposals it received, and otherwise $T$ declines all proposals it received.
We now set 

\[V(F_{j+1}) = V(F_j) \setminus \left( \bigcup_{\substack{v \in V^{propose}_j, \\ \text{ the proposal of $v$ was declined}}} V(T_v) \right).\]

Each node in $V(F_{j+1}) \setminus V^{propose}_i$ has the same parent in $F_{j+1}$ and $F_j$, or is a root in both $F_{j+1}$ and $F_j$.
Each node in $V(F_{j+1}) \cap V^{propose}_j$, i.e., each node whose proposal got accepted by some red tree $T$ in $F_j$, changes its parent to be an arbitrary neighboring node in the tree $T$. 
Note that if a red tree $T$ decides to grow, then the corresponding tree in $F_{j+1}$ contains at least $\left(1 + \frac{1}{2b}\right)|V(T)|$ vertices. Moreover, if $T$ does not decide to grow, then $T$ is also a tree in $F_{j+1}$ and is not neighboring with any blue tree in $F_{j+1}$. This follows from the fact that each blue node neighboring a red tree either becomes a red node or gets deleted. 

We now have fully specified how the rooted forests $F_0, F_1, \ldots, F_t$ are computed and recall that in the end we set $V_{i+1} = V(F_t)$ and $Q_{i+1}$ is the set of roots of the forest $F_t$.

\subsection{Analysis}
\label{sec:analysis_phase}

For each $j \in \{0,1,\ldots,t\}$ and $u \in V(F_j)$, we define $d_j(u)$ as the length of the path from $u$ to its root in $F_j$. Note that as $F_0$ is a BFS forest, for any neighboring nodes $w,v \in V(F_0)$ it holds that $d_0(w) \leq d_0(v) + 1$.

\begin{claim}[Ruling Claim]
\label{claim:distance_to_root}
For every $i \in \{0,1,\ldots,t\}$, the following holds:

\begin{enumerate}[leftmargin = 4cm]
    \item [Blue Property:\hspace{1em}] Every blue node in $F_j$ satisfies $d_j(u) = d_0(u)$.
    \item [Red Property:\hspace{1em}]  Every red node in $F_j$ satisfies $d_j(u) \leq d_0(u) + 2j$.
\end{enumerate}
In particular, this implies that Invariant (I) is preserved.
\end{claim}
\begin{proof}
The blue property directly follows from the fact that for any blue node the path to its root in $F_j$ is the same as the path to its root in $F_0$.
We prove the red property by induction on $j$. The base case $j = 0$ trivially holds.

For the induction step, consider an arbitrary $j \in \{0,1,\ldots,t-1\}$.
We show that the statement holds for $j+1$ given that it holds for $j$.

Consider an arbitrary red node $u$ in $F_{j+1}$. We have to show that $d_{j+1}(u) \leq d_0(u) + 2(j+1)$. If $u$ is also a red node in $F_j$, then we can directly use induction.
Hence, it remains to consider the case $u$ is a blue node in $F_j$. 

In that case, there exists a node $v \in V^{propose}_j$ such that $u \in V(T_v)$ and the proposal of $v$ was accepted. In particular, $v$'s parent in $F_{j+1}$ is some neighboring node $w$ which is part of some red tree in $F_j$ (see \cref{fig:rehanging}). 

The path from $u$ to its root $r$ in $F_{j+1}$ can be decomposed into a path from $u$ to $v$, an edge from $v$ to $w$ and a path from $w$ to its root $r$.

The path from $u$ to $v$ in $F_{j+1}$ is the same as the path from $u$ to $v$ in $F_0$ and therefore of length $d_0(u) - d_0(v)$. 
The path from $w$ to $r$ in $F_{j+1}$ is the same as the path from $w$ to $r$ in $F_j$ and therefore has a length of $d_j(w)$ with $d_j(w) \leq d_0(w) + 2j$ according to the induction hypothesis. 
Moreover, we noted above that because $w$ and $v$ are neighbors, we have $d_0(w) \leq d_0(v) + 1$.
Hence, we can upper bound the length of the path from $u$ to its root in $F_{j+1}$ by

\[d_{j+1}(u) \leq \left( d_0(u) - d_0(v)\right) + 1 + \left( d_0(w) + 2j\right) \leq d_0(u) + 2(j+1)\]
which finishes the induction proof. It remains to prove the last part of the claim. To that end, assume that the ruling invariant is satisfied for $i$, i.e., $Q_i$ is $R_i$-ruling in $G[V_i]$ for $R_i = i \cdot O(\log^2 n)$. Then, every node $u$ in $V(F_t) = V_{i+1}$ satisfies

\[d_{G[V_{i+1}]}(Q_{i+1},u) \leq d_t(u) \leq d_0(u) + 2t \leq i \cdot O(\log^2 n) + O(\log^2 n) = (i+1)O(\log^2 n)\]

and therefore the ruling invariant is satisfied for $i+1$.
\end{proof}

\begin{figure}
    \centering
    \includegraphics{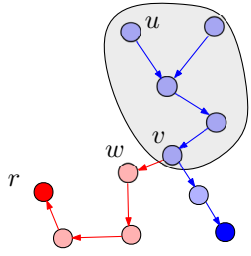}
    \caption{The figure shows the situation in the proof of \cref{claim:distance_to_root}. The path from $u$ to $r$ splits into three parts: from $u$ to $v$, then to $w$, then to $r$. The length of each part is upper bounded separately. }
    \label{fig:rehanging}
\end{figure}

\begin{claim}(Separation Claim)
\label{claim:tree_size}
No red node in $F_t$ is neighboring a blue node in $F_t$. 
In particular, this implies that Invariant (II) is preserved.
\end{claim}
\begin{proof}
    We observed during the algorithm description that each red tree that decides to grow grows by at least a $(1 + \frac{1}{2b})$-factor in a given step. Our choice of $t = 2b^2$ implies that
    
    \[\left(1 + \frac{1}{2b}\right)^{t} = \left( \left(1 + \frac{1}{2b}\right)^{2b} \right)^{(t/2b)} > 2^{t/2b} = 2^b \geq n,\]
    
   and therefore each tree eventually stops growing. However, once a tree decides not to grow, it is not neighboring with any blue node and therefore no red node in $F_t$ is neighboring a blue node in $F_t$.
    In particular, this implies that each connected component of $G[V_{i+1}] = G[V(F_t)]$ either entirely consists of blue nodes in $F_t$ or entirely consists of red nodes in $F_t$. As the $(i+1)$-th bit of the identifier of each red root in $F_t$ is $0$ and the $(i+1)$-th bit of the identifier of each blue root in $F_t$ is $1$, we get that each connected component of $G[V_{i+1}]$ either contains no node in $Q_{i+1}$ with the $(i+1)$-th bit of the identifier being $0$ or no node in $Q_{i+1}$ with the $(i+1)$-th bit of the identifier being $1$, which implies that the separation invariant is preserved.
\end{proof}

\begin{claim}[Deletion Claim]
\label{claim:deletion}
It holds that $|V_{i+1}| = |V(F_t)| \geq \left(1 - \frac{1}{2b}\right)|V(F_0)| \geq |V_i| - \frac{|V|}{2b}$.
In particular, this implies that Invariant (III) is preserved.
\end{claim}
\begin{proof}
A node $u$ got deleted in step $i$, i.e., $u \in V(F_j) \setminus V(F_{j+1})$, because of some tree $T$ in $F_j$ which decided to stop growing, as

\[\sum_{v \in V_j^{propose} \colon \text{$v$ proposes to $T$}} |V(T_v)| < \frac{|V(T)|}{2b}.\]

We blaim this tree $T$ for deleting $u$. Note that $T$ only receives blaim in step $j$ and at most $\frac{|V(T)|}{2b}$ deleted nodes blaim $T$. During the algorithm description, we observed that $T$ is not neighboring any blue node in $F_{j+1}$ and therefore $T$ is also a tree in $F_t$. Hence, each deleted node in $V(F_0) \setminus V(F_t)$ can blaim one tree $T$ in $F_t$ for being deleted in such a way that each such tree gets blaimed by at most $\frac{1}{2b}|V(T)|$ nodes, which directly proofs the claim.
\end{proof}

\begin{proof}[Proof of \cref{thm:clustering_theorem}]
The algorithm has $O(\log n)$ phases, with each phase consisting of $O(\log^2 n)$ steps. It directly follows from the ruling claim that each step can be executed in $O(\log^3 n)$ \congest rounds. Hence, we can compute $V'$ in $O(\log^6 n)$ \congest rounds, which together with the previous discussion finishes the proof of \cref{thm:clustering_theorem}. 
\end{proof}

\section{Acknowledgments}
We want to thank Mohsen Ghaffari for many valuable suggestions. 

\bibliographystyle{alpha}
\bibliography{ref}

\end{document}